\documentclass[11pt,a4paper]{article}
  \usepackage{authblk}
  \usepackage{cite}
  \usepackage{graphicx}
  \usepackage{latexsym, amsmath, amsthm, amssymb}
\usepackage[usenames,dvipsnames]{color}

\usepackage{latexsym, amsmath, amssymb}
\usepackage[usenames,dvipsnames]{color}
\usepackage{comment}

\theoremstyle{plain}
\newtheorem{theorem}{Theorem}
\newtheorem{proposition}[theorem]{Proposition}
\newtheorem{lemma}[theorem]{Lemma}
\newtheorem{corollary}[theorem]{Corollary}

\theoremstyle{definition}
\newtheorem{definition}{Definition}
\newtheorem*{remark}{Remark}

  \title{An approach to comparing Kolmogorov-Sinai and permutation entropy}

  \author[1,2]{Valentina A.~Unakafova\thanks{Corresponding address: Institute of Mathematics, University of L\"ubeck,
	     Ratzeburger Alley 160, Building 64, 23562 L\"ubeck, Germany. Tel.: +49 451 500 3165; fax: +49 451 500 3373. e-mail: unakafova@math.uni-luebeck.de (Valentina A.~Unakafova)}}
  \author[1,2]{Anton M.~Unakafov}
  \author[1]{Karsten Keller}
  \affil[1]{Institute of Mathematics, University of L\"ubeck}
  \affil[2]{Graduate School for Computing in Medicine and Life Sciences, University of L\"ubeck}
  \date{January 4, 2013}
  \begin{document}
  \maketitle

  \begin{abstract}
  \noindent 
  In this paper we discuss the relationship between permutation entropy and Kolmogorov-Sinai entropy in the one-dimensional case. 
  For this, we consider partitions of the state space of a dynamical system using ordinal patterns of order $(d+n-1)$ on the one hand, 
  and using $n$-letter words of ordinal patterns of order $d$ on the other hand. 
  The answer to the question of how different these partitions are provides an approach to comparing the entropies.
  \end{abstract}

\section{Introduction}\label{intro}

In this paper we discuss the relationship between the permutation entropy, introduced by Bandt and Pompe \cite{BandtPompe2002}, and
the well-known Kolmogorov-Sinai entropy (KS entropy).
A significant result in this direction, given by Bandt, Keller, and Pompe \cite{BandtKellerPompe2002}, is equality of both entropies for piecewise strictly monotone interval maps.
For many dynamical systems KS entropy has been shown to be not larger than permutation entropy \cite{KellerSinn2009,KellerSinn2010,Keller2011}.
Amig\'{o} et al. have proved equality of KS entropy and permutation entropy for a slightly different concept
of permutation entropy \cite{AmigoKennelKocarev2005,Amigo2012} (for a detailed discussion see \cite{Amigo2010}).

The representation of KS entropy on the basis of ordinal partitions given in \cite{KellerSinn2009,KellerSinn2010,Keller2011} allows to relate permutation entropy
and KS entropy. Roughly speaking, ordinal partitions classify the points of the state space according to the order types
(ordinal patterns) of their orbits.
The next step for better understanding the relationship of the entropies is to answer to
the question of how much more information ordinal patterns of order $(d+n-1)$ provide than
$n$ overlapping ordinal patterns of order $d$ \cite{KellerUnakafovUnakafova2012}.
Here we specialize the considerations in \cite{KellerUnakafovUnakafova2012} to the case of one-dimensional dynamical system.
At this level of research we do not have conclusive results, but we present some new ideas in this direction.

\subsection{Preliminaries}\label{preliminaries}

 Throughout the paper, $(\Omega,\mathbb{B}(\Omega),\mu,T)$ is a measure-preserving dynamical system,
 where $\Omega$ is an interval in $\mathbb{R}$,
 $\mathbb{B}(\Omega)$ is the Borel sigma-algebra on it,
 $\mu:\mathbb{B}(\Omega) \rightarrow [0,1]$ is a probability measure with $\mu(\{\omega\})=0$ for all $\omega\in\Omega$,
 and $T: \Omega \hookleftarrow$ is a $\mathbb{B}(\Omega)$-$\mathbb{B}(\Omega)$-measurable $\mu$-preserving transformation,
 i.e. $\mu(T^{-1}(B))=\mu(B)$ for all $B \in \mathbb{B}(\Omega)$.

 The {\it(Shannon) entropy} of a finite partition ${\cal P}=\{P_1,P_2,\ldots,P_l\} \subset \mathbb{B}(\Omega)$ of $\Omega$ with respect to $\mu$ is defined by
 \begin{equation*}
  H({\cal P})= - \sum_{P \in {\cal P}} \mu(P)\ln\mu(P)
 \end{equation*}
 (with $0\ln 0 := 0$).

  The alphabet $A=\{ 1,2, \ldots, l \}$ corresponding to a finite partition ${\cal P}=\{P_1,P_2,\ldots,\linebreak P_l\}$ provides
  words ${a_1 a_2 \ldots a_n}$ of given length $n$, and the set $A^n$ of all such words provides a partition
  ${\cal P}_n$ of $\Omega$ into the sets
  \begin{equation*}
	P_{a_1 a_2 \ldots a_n}=\{\omega \in P_{a_1}, T(\omega) \in P_{a_2}, \ldots, T^{\circ n-1}(\omega) \in P_{a_n}\}.
  \end{equation*}
  Here $T^{\circ t}$ denotes the $t$-th iterate of $T$.

  The {\it Kolmogorov-Sinai entropy (KS entropy)} and the {\it permutation entropy} of $T$ are defined by
  \begin{equation*}
       h_{\mu}(T) = \sup_{{\cal P}\ \text{finite partition }} \lim_{n \to \infty } \frac{H({\cal P}_n )}{n}
  \end{equation*}
  and
  \begin{equation*}
       h_{\mu}^{*}(T)=\varlimsup_{d \rightarrow \infty}\frac{H ({\cal P}(d)) }{d},
  \end{equation*}
  respectively, where ${\cal P}(d)$ is the ordinal partition we will consider in Section \ref{FromOrdinalPatternsToWords}.

  It was shown in \cite{KellerSinn2009,KellerSinn2010,Keller2011} that for many cases ordinal partitions characterize
  the KS entropy of $T$ in the following way:
  \begin{equation}\label{KS_representation}
    h_\mu(T)=\lim_{d \rightarrow \infty}\lim_{n \rightarrow \infty}\frac{H({\cal P}(d)_n)}{n}.
  \end{equation}
  (The partition ${\cal P}(d)_n$ given ${\cal P}(d)$ fits into the general definition of ${\cal P}_n$ given ${\cal P}$ as defined above.)

\subsection{Relationship between KS and permutation entropy}\label{KS_PE}

For the following discussion, recall the main result from \cite{KellerUnakafovUnakafova2012}.

\begin{theorem}
	The following statements are equivalent for $h_\mu(T)$ satisfying \eqref{KS_representation}:
	\begin{enumerate}
		\item [(i)]  $h_\mu(T) =h_{\mu}^\ast(T)$.
		\item [(ii)] For each $\varepsilon > 0$ there exists some $d_\varepsilon \in \mathbb{N}$ such that
			     for all $d \geq d_\varepsilon$ there is some $n_d \in \mathbb{N}$ with
			     \begin {equation}
				  H({\cal P}(d + n - 1)) - H({\cal P}(d)_n) < (n - 1)\varepsilon \text{ for all }n \geq n_d.
				  \label{MainInequality}
			     \end {equation}	
	\end{enumerate}
\end{theorem}
The purpose of the following discussion is to compare the partitions ${\cal P}( d + n - 1 )$ and ${\cal P}( d )_n$ and to answer the question
under what assumptions (ii) in Theorem \ref{Theorem1} holds and, more generally,
in what extent these partitions differ with increasing $d$ and $n$.

Let us define $V_d \subset \Omega$ as
\begin{align}\label{BadSet}
     \nonumber
	 V_d\ =\ &\{\omega \mid \omega < T^{\circ d}(\omega),\, T^{\circ l}(\omega) \notin (\omega,T^{\circ d}(\omega))
	 \text{ for all } l = 1, \ldots, d-1 \}\\	
	 \cup\ &\{\omega \mid \omega \geq T^{\circ d}(\omega),\, T^{\circ l}(\omega) \notin [T^{\circ d}(\omega),\omega]
	 \text{ for all } l = 1, \ldots, d-1 \}.
\end{align}
The sets $V_{d+1},\ldots, V_{d+n-1}$, more closely considered in Section \ref{FromOrdinalPatternsToWords},
allow to describe all elements of the partition ${\cal P}(d+n-1)$, which are proper subsets of some elements of the partition ${\cal P}(d)_n$.
We are interested in showing that the sets $V_d$ are small in a certain sense.

Recall that $T$ is said to be {\it mixing} or {\it strong-mixing} if for every $A,B \in \mathbb{B}(\Omega)$
\begin{equation*}
 \lim_{n \rightarrow \infty}\mu(T^{-\circ n}A \cap B)=\mu(A)\mu(B).
\end{equation*}

\begin{theorem}\label{Theorem1}
  If $T$ is mixing, then for all $\varepsilon > 0$ there exists some $d_\varepsilon$ such that for all $d > d_\varepsilon$
  \begin{equation}
       \mu(V_d)<\varepsilon.
  \end{equation}
\end{theorem}
Theorem \ref{Theorem2} provides a tool for comparing
``successive'' partitions ${\cal P}(d+1)_{n-1}$ and ${\cal P}( d )_n$.

\begin{theorem}\label{Theorem2}
  For all $n\in {\mathbb N}\setminus\{1\}$ and $d \in \mathbb{N}$ it holds
  \begin {equation}\label{IntermediateResult}
	H({\cal P}(d+1)_{n-1}) - H({\cal P}(d)_n) \leq \ln2(n-1)\mu(V_{d+1}).
  \end {equation}
\end{theorem}
Putting together Theorem \ref{Theorem1} and Theorem \ref{Theorem2}, one gets a more explicit variant of \eqref{IntermediateResult}:

\begin{corollary}\label{Corollary1}
  If $T$ is mixing, then for all $\varepsilon > 0$ there exists some $d_\varepsilon \in \mathbb{N}$ such that
  for all $d \geq d_\varepsilon,n \in \mathbb{N}\setminus\{1\}$ it holds
  \begin {equation*}
	H({\cal P}(d+1)_{n-1}) - H({\cal P}(d)_n) < (n-1)\varepsilon.
  \end {equation*}
\end{corollary}
Coming back to the partitions ${\cal P}( d + n - 1 )$ and ${\cal P}( d )_n$, in Section \ref{Partitions2}
we obtain the following upper bound for $H({\cal P}(d + n - 1)) - H({\cal P}(d)_n)$:
\begin{equation}
      H({\cal P}(d + n - 1)) - H({\cal P}(d)_n) \leq \ln 2 \sum^{n-1}_{i=1}(n-i)\mu(V_{d+i})
      \label{OverBound}
\end{equation}
(compare with \eqref{MainInequality}).

Being the main results of our paper, Theorem \ref{Theorem1}, Theorem \ref{Theorem2}, and Corollary \ref{Corollary1} shed some new light
on the general problem of equality between Kolmogorov-Sinai and permutation entropy in the one-dimensional case.

Section \ref{FromOrdinalPatternsToWords} gives the detailed description of $n$-letter words with ordinal patterns of order $d$ as letters,
of ordinal patterns themselves and their connection to the sets $V_d$.
In Section \ref{Partitions1} we focus on the partitions ${\cal P}(d+1)_{n-1}$ and ${\cal P}( d )_n$ and prove Theorem \ref{Theorem2}.
Section \ref{Partitions2} is devoted to the relation of the partitions ${\cal P}( d + n - 1 )$ and ${\cal P}( d )_n$ and provides \eqref{OverBound}.
Finally, we prove  Theorem \ref{Theorem1} in Section \ref{Proof}.

\section {From ordinal patterns to words}\label{FromOrdinalPatternsToWords}

Let us recall the definition of ordinal patterns.
\begin{definition}
    Let $\Pi_d$ be the set of permutations of the set $\lbrace 0, 1, 2, ..., d\rbrace$ for $d \in \mathbb{N}$. Then
    the real vector $(x_0, x_1, ..., x_d)\in {\mathbb R}^{d+1}$
    has {\it ordinal pattern} $\pi = (r_0, r_1,\ldots, r_d) \in \Pi_d$ of order $d$ if
    \begin {equation*}
	x_{r_0} \geq x_{r_1} \geq ... \geq x_{r_d}
    \end {equation*}
    and
    \begin {equation*}
	r_{l-1} > r_{l} \text{ in the case } x_{r_{l-1}} = x_{r_{l}}.
    \end {equation*}	
\end{definition}

We divide now the state space into sets of points having similar dynamics from the ordinal viewpoint.
\begin{definition}
  For $d \in \mathbb{N}$, the partition ${\cal P}(d) = \{P_{\pi} \mid \pi \in \Pi_d \} \text{ with }$
  \begin{equation*}
    P_{\pi} = \{ \omega \in \Omega \mid (T^{\circ d}(\omega), T^{\circ d-1}(\omega), \ldots, T(\omega), \omega) \text{ has ordinal pattern } \pi \}
  \end{equation*}
  is called {\it ordinal partition of order $d$} with respect to $T$.
\end{definition}
A finer partition is obtained by considering more than one successive ordinal pattern.
\begin{definition}
    We say, that a real vector $(x_0, x_1, ..., x_{d+n-1})\in {\mathbb R}^{d+n}$ has {\it$(n,d)$-word}
    $\pi_1 \pi_2 \ldots \pi_n$ if\\
    \begin {equation*}
	(x_i,x_{i+1},\ldots,x_{i+d}) \text{ has ordinal pattern } \pi_{i+1} \in \Pi_d \text{ for } i=0,1,\ldots,n-1.
    \end {equation*}	
\end{definition}
The partition ${\cal P}(d)_n$ associated to the collection of $(n,d)$-words consists of the sets
\begin{equation*}
   P_{\pi_1 \pi_2 \ldots \pi_n}=\{\omega \in P_{\pi_1}, T(\omega) \in P_{\pi_2}, \ldots, T^{\circ n-1}(\omega) \in P_{\pi_{n}}\}, \pi_1,\pi_2, \ldots, \pi_{n} \in \Pi_d.
\end{equation*}

Figure \ref{fig1} illustrates a segment $(\omega, T(\omega),\ldots,T^{\circ 5}(\omega))$ of some orbit (a)
and the corresponding $(5,1)$-, $(4,2)$-, $(3,3)$-, $(2,4)$- and $(1,5)$-words (b).

\begin{figure}[h]
  \begin{minipage}{0.49\linewidth}
    \includegraphics[scale=0.5]{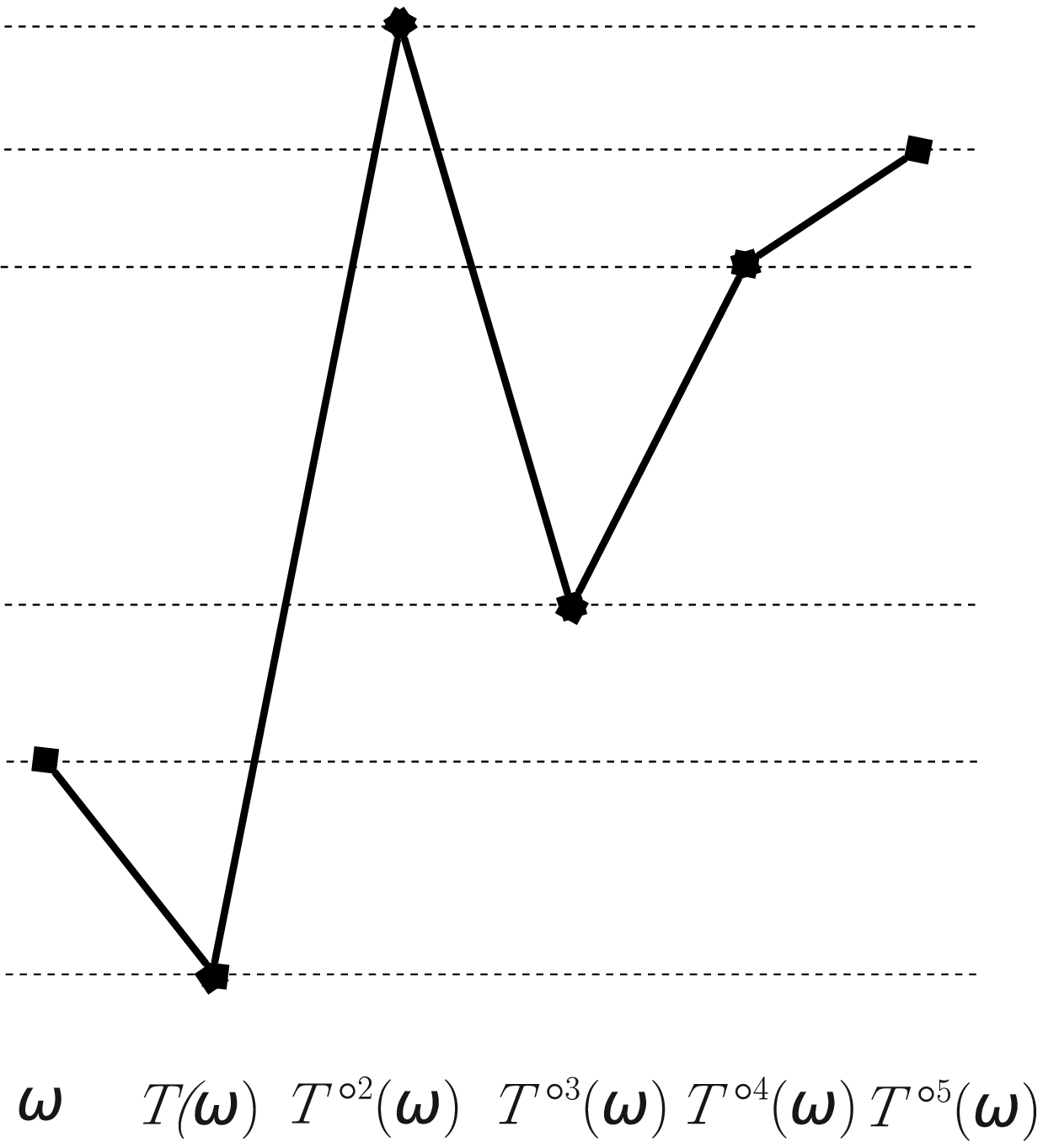}
    \begin{center}
     (a)
    \end{center}
  \end{minipage}
  \hfill
  \begin{minipage}{0.49\linewidth}
    \includegraphics[scale=0.5]{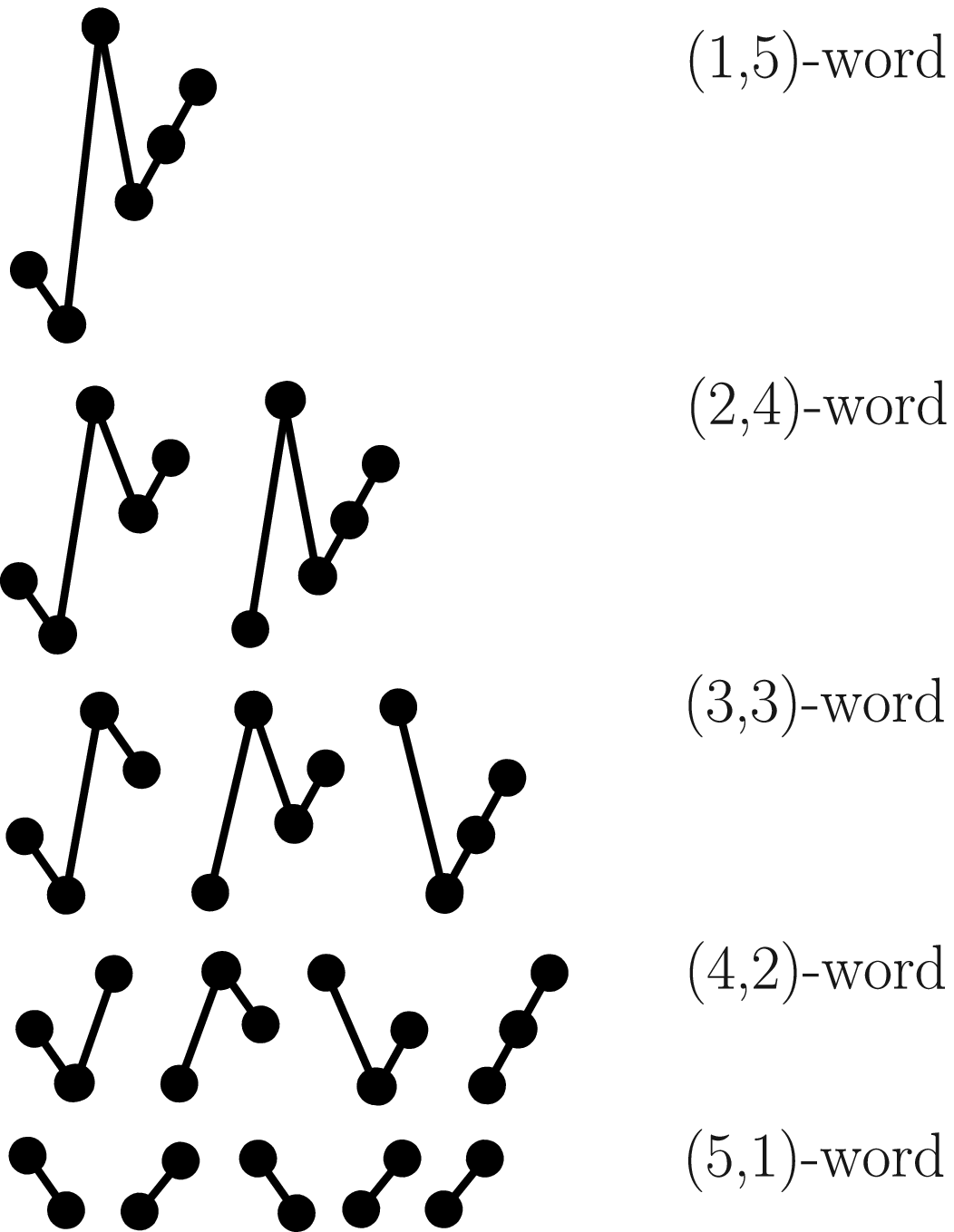}
    \begin{center}
     (b)
    \end{center}
 \end{minipage}
  \caption{Representation of the segment of the orbit (a) by $(n,d)$-words (b)}
  \label{fig1}
\end{figure}

Upon moving from $(1,5)$- to $(5,1)$-words one loses some information about the ordering of the iterates of $T$.
For example, the $(3,3)$-word determines the relation
\begin{equation*}
 \omega < T^{\circ 3}(\omega),
\end{equation*}
but in the $(4,2)$-word this relation is already lost.
It either holds $\omega \geq T^{\circ 3}(\omega)$ or $\omega < T^{\circ 3}(\omega)$.

On the other hand, one does not lose the relation
\begin{equation*}
 \omega < T^{\circ 4}(\omega)
\end{equation*}
when moving from the $(2,4)$-word to the $(3,3)$-word, although $\omega$ and $T^{\circ 4}(\omega)$ are in different patterns of the $(3,3)$-word.
The reason for this is the existence of the intermediate iterate $T^{\circ 3}(\omega)$ with
\begin{equation*}
      \omega < T^{\circ 3}(\omega)<T^{\circ 4}(\omega).
\end{equation*}

More generally, if there is some intermediate iterate $T^{\circ l}(\omega)$ with $\omega<T^{\circ l}(\omega)<T^{\circ d+1}(\omega)$ or $T^{\circ d+1}(\omega) \leq T^{\circ l}(\omega) \leq \omega$,
the relation between $\omega$ and $T^{\circ d+1}(\omega)$ is not lost upon moving from $(1,d+1)$- to $(2,d)$-words, and is lost otherwise.
Therefore, the set $V_{d+1}$ (see \eqref{BadSet}) consists of all $\omega$, for which the relation between $\omega$ and $T^{\circ d+1}(\omega)$ is lost
upon moving from $(1,d+1)$- to $(2,d)$-words. More precisely, the set $V_{d+1}$ is a union of the sets of the partition ${\cal P}(d+1)$ that
 are proper subsets of some sets of the partition ${\cal P}(d)_2$.

Figure \ref{fig2} illustrates $\omega,T^{\circ 2}(\omega) \in V_3$ for our example.

\begin{figure}[h]

  \begin{minipage}{0.49\linewidth}
    \includegraphics[scale=0.5]{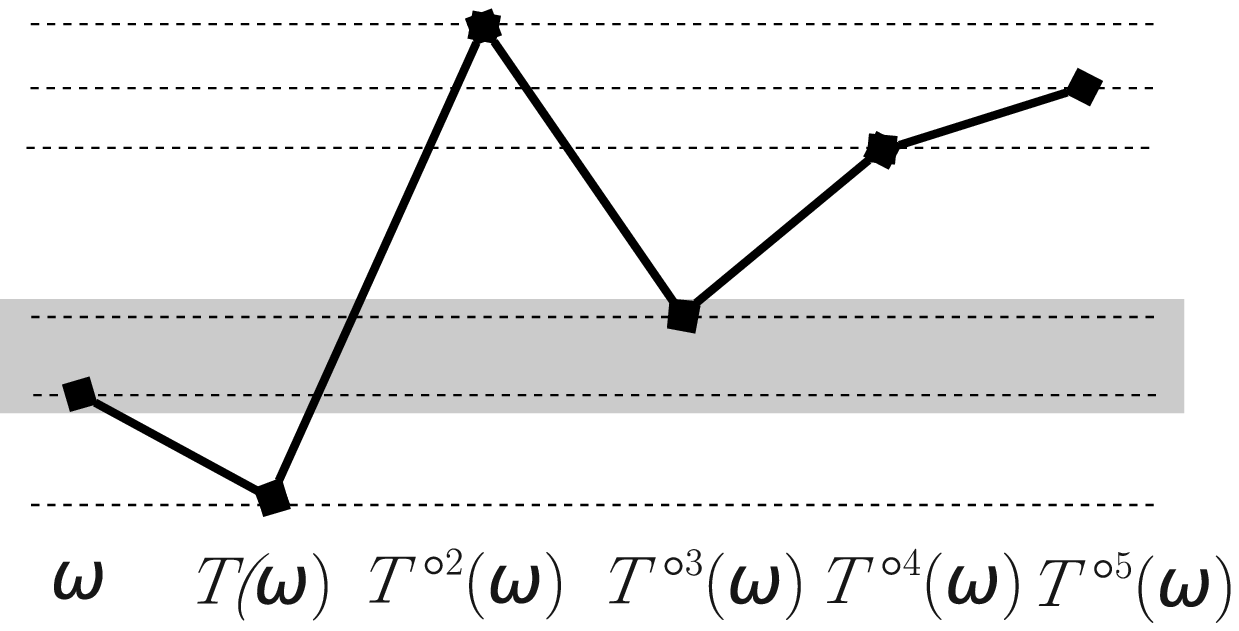}
    \begin{center}
     (a)
    \end{center}
  \end{minipage}
  \hfill
  \begin{minipage}{0.49\linewidth}
    \includegraphics[scale=0.5]{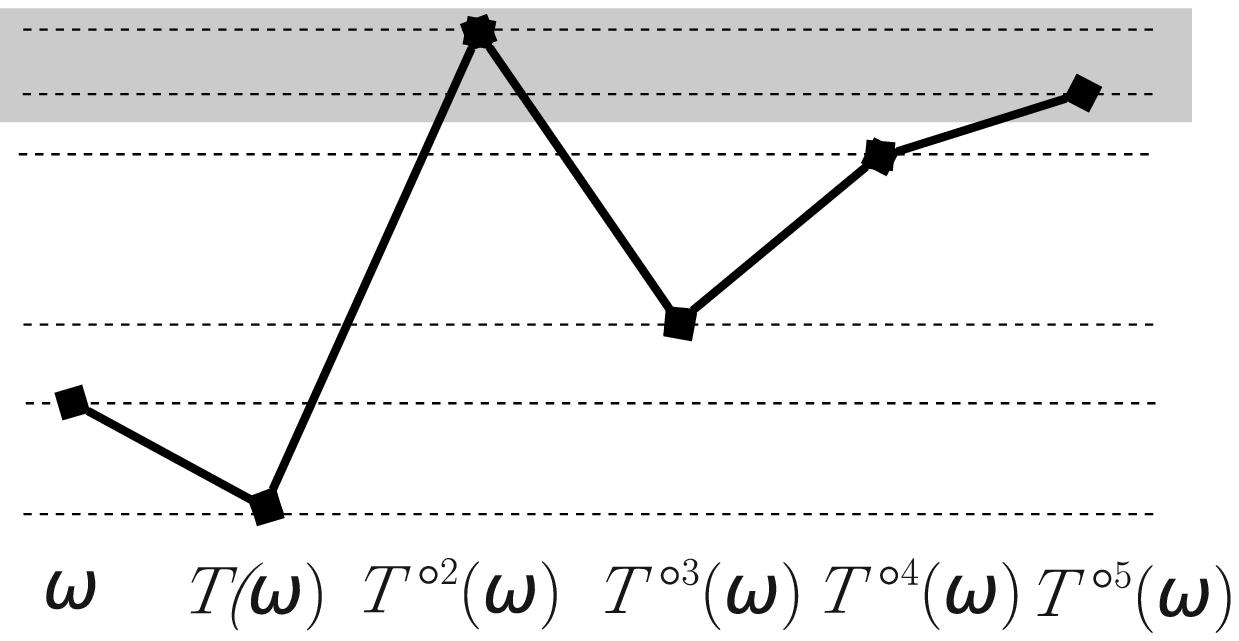}
    \begin{center}
     (b)
    \end{center}
 \end{minipage}
  \caption{$\omega\in V_3$ (a), $T^{\circ 2}(\omega) \in V_3$ (b)}
  \label{fig2}
\end{figure}

In the following section we compare the partitions ${\cal P}( d )_n$ and ${\cal P}( d + 1 )_{ n - 1 }$
by means of the set $V_{d+1}$.

\section{The partitions ${\cal P}( d + 1 )_{ n - 1 }$ and ${\cal P}( d )_n$}\label{Partitions1}

Upon moving from $(n-1,d+1)$- to $(n,d)$-words, for $i = 0,1,\ldots,n-2$ the relation
between $T^{\circ i}(\omega)$ and $T^{\circ d+i+1}(\omega)$ is lost iff $T^{\circ i}(\omega) \in V_{d+1}$.
Therefore, if $V_{d+1}\neq\emptyset$, then
the partition ${\cal P}( d + 1 )_{ n - 1 }$ is properly finer than the partition ${\cal P}( d )_n$. The following is valid:

\begin{proposition}\label{Proposition1}
  Given $P \in {\cal P}(d)_n$, let $k = \# \{l \in \{0, 1, \ldots, n-2\} \mid P \subset T^{- \circ l}(  V_{d+1} ) \}$.  
  Then there exist $2^k$ sets $P_1, P_2, \ldots, P_{2^k} \in {\cal P}(d+1)_{n-1}$ with
  \begin{equation*}
      P_1 \cup P_2 \cup \ldots \cup P_{2^k} = P.
  \end{equation*}
\end{proposition}
\begin{proof}
  Consider some $P \in {\cal P}(d)_n$ and the corresponding $(n,d)$-word.
  Since the $(n,d)$-word determines the same dynamics for all $\omega \in P$,
  for $l=0,1,\ldots,n-2$ it holds either
  \begin{equation}\label{SubsetOfBadSet}
      P \subset T^{-\circ l}(V_{d+1})
  \end{equation}
  or
  \begin{equation}\label{NotSubsetOfBadSet}
      P \cap T^{-\circ l}(V_{d+1}) = \emptyset.
  \end{equation}
  For each $l$ with \eqref{SubsetOfBadSet} and all $\omega \in P$, either $T^{\circ l}(\omega) < T^{\circ d+l+1}(\omega)$ or $T^{\circ d+l+1}(\omega) \leq T^{\circ l}(\omega)$ providing a division of $P$ into two subset. We are done since there are exactly $k$ such divisions.
\end{proof}

Figure \ref{fig3} illustrates Proposition \ref{Proposition1}.
\begin{figure}[ht]
  \centering
  \includegraphics[scale=0.5]{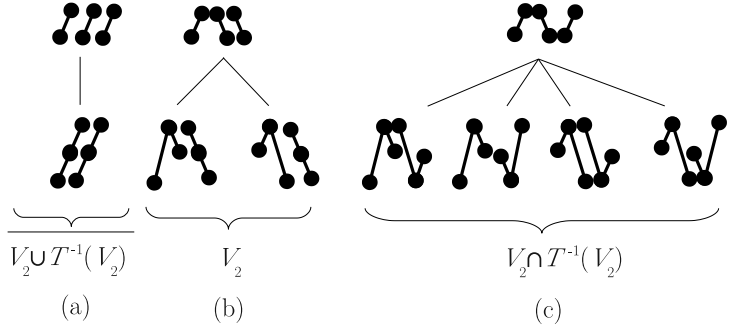}
  \caption{From $(3,1)$- to $(2,2)$-words. $\overline{V_2 \cup T^{-1}(V_2)}$ in (a) stands for the complement of $V_2 \cup T^{-1}(V_2)$}
  \label{fig3}
\end{figure}
For $\omega \notin V_2 \cup T^{-1}(V_2)$ the obtained $(3,1)$-word is not divided and contains the same information about the ordering as $2^0=1$ $(2,2)$-word (a),
for $\omega \in V_2$ the $(3,1)$-word is divided into $2^1=2$ $(2,2)$-words (b) and
for $\omega \in V_2 \cap T^{-1}(V_2)$ the $(3,1)$-word is divided into $2^2=4$ $(2,2)$-words (c).

Let $k(P)$ be determined as in Proposition \ref{Proposition1} for each $P \in {\cal P}(d)_n$.
Since for each $P$ it holds either \eqref{SubsetOfBadSet} or \eqref{NotSubsetOfBadSet}, it follows
\begin{equation}\label{ForFirstBound}
	\sum_{j=0}^{n-2}\mu(T^{-\circ j}(V_{d+1}))=\sum_{j=0}^{n-2}\sum_{P \in {\cal P}(d)_n}\mu(T^{-\circ j}(V_{d+1})\cap P)
	=\sum_{P \in {\cal P}(d)_n}k(P)\mu(P).
\end{equation}
Therefore, by Proposition \ref{Proposition1} and \eqref{ForFirstBound} one obtains an upper bound for $H({\cal P}(d+1)_{n-1})-H({\cal P}(d)_n)$ in the following way:
\begin{align}\label{FirstBound}
  \nonumber
  H({\cal P}(d+1)_{n-1}) - H({\cal P}(d)_{n}) &\leq \sum_{P \in {\cal P}(d)_n} \left( \mu(P) \ln \mu(P) - 2^{k(P)} \frac{\mu(P)}{2^{k(P)}} \ln\frac{\mu(P)}{2^{k(P)}} \right) \\
  = \ln 2 \sum_{P \in {\cal P}(d)_n} k(P) \mu(P) &= \ln 2 \sum_{j=0}^{n-2}\mu(T^{- \circ j}(V_{d+1})) = \ln 2 (n-1) \mu(V_{d+1}).
\end{align}
Inequality \eqref{FirstBound} provides the proof of Theorem \ref{Theorem2}.

\section{The partitions ${\cal P}( d )_n$ and ${\cal P}( d + n - 1 )$}\label{Partitions2}
Here we move from $(1,d+n-1)$-words (i.e. ordinal patterns of order $(d+n-1)$) to $(n,d)$-words.
At this point we cannot definitely say into how many $(n,d)$-words a $(1,d+n-1)$-word
is divided in dependence on the sets $V_{d+1},\ldots,V_{d+n-1}$.

Let us give an example. Figure \ref{fig4} illustrates a $(3,1)$-word with the same information as in the $(1,3)$-word (a),
other two $(3,1)$-words are divided into three and five $(1,3)$-words ((b) and (c), respectively).
\begin{figure}[ht]
  \centering
  \includegraphics[scale=0.5]{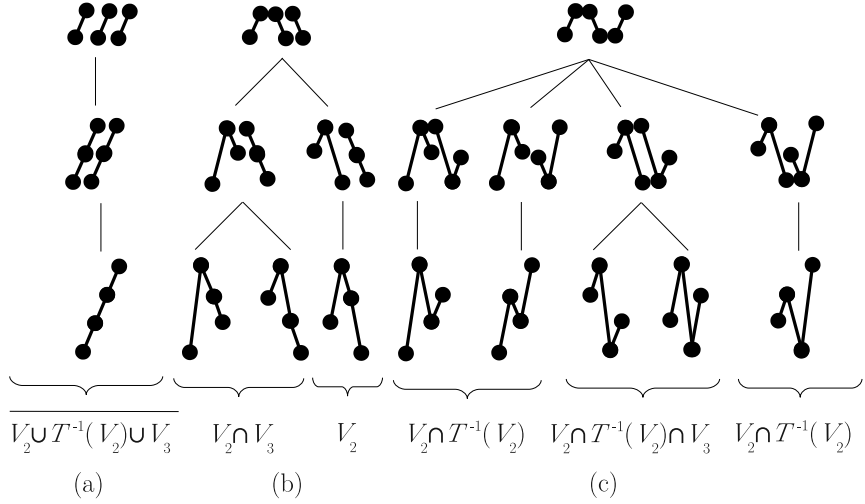}
  \caption{From $(3,1)$- to $(1,3)$-words. $\overline{V_2 \cup T^{-1}(V_2) \cup V_3}$ in (a) stands for the complement of $V_2 \cup T^{-1}(V_2) \cup V_3$}
  \label{fig4}
\end{figure}

One obtains an upper bound for $H({\cal P}(d+n-1)) - H({\cal P}(d)_n)$ by successive application of \eqref{FirstBound}:
\begin{align}\label{UpperBound}
    \nonumber
    H({\cal P}(d+n-1)) - H({\cal P}(d)_n) &= \sum_{i=1}^{n-1}( H({\cal P}(d+n-i)_i) - H({\cal P}(d+n-i-1)_{i+1}) )\\
    &\leq \ln 2 \sum_{i=1}^{n-1} i\,\mu(V_{d+n-i}) = \ln 2 \sum^{n-1}_{i=1}(n-i)\,\mu(V_{d+i}).
\end{align}
Comparing \eqref{MainInequality} and \eqref{UpperBound} it is natural to ask how fast the measure of the set $V_d$
decreases with increasing $d$. This question is the subject of current research.

\section{Proof of Theorem \ref{Theorem1}}\label{Proof}
In the following, we assume that $T$ is strong-mixing, however, some parts of the proof need only the weaker assumption of ergodicity, as we will indicate.
\begin{lemma}\label{Lemma1}
  Let $T$ be ergodic.
  Given an interval $A \subset \Omega$ and $d\in {\mathbb N}\setminus\{1\}$, let $\widetilde{V}_d=\widetilde{V}_d(A)$ be the set of points $\omega \in A$ for which at least one of two following conditions holds:
  \begin {equation} \label{LeftCondition}
     T^{\circ l}(\omega) \notin \{a\in A\,|\,a<\omega\} \text{ for all } l = 1,...,d-1,
  \end {equation}
  \begin {equation} \label{RightCondition}
     T^{\circ l}(\omega) \notin \{a\in A\,|\,a>\omega\} \text{ for all } l = 1,...,d-1.
  \end {equation}

  Then for all $\varepsilon > 0$ there exists some $d_\varepsilon \in {\mathbb N}$ such that $\mu(\widetilde{V}_d) < \varepsilon$ for all $d > d_\varepsilon$.
\end{lemma}
\begin{proof}
  Let $\widetilde{V}^L_d$ be a set of points $\omega$ satisfying \eqref{LeftCondition}. Then it is sufficient to show $\mu(\widetilde{V}^L_d)<\frac{\varepsilon}{2}$ for the corresponding $d$ since
  for points satisfying \eqref{RightCondition} the proof is completely resembling.

  Consider a partition $\{B_i\}_{i=1}^\infty$ of $A$ into intervals $B_i$ with the following properties:
  \begin{enumerate}
	\item[(i)] $\mu(B_i) = \frac{\mu(A)}{2^i}$ for all $i \in {\mathbb N}$,
	\item[(ii)] for all $i < j$, and for all $\omega_1 \in B_i,\omega_2 \in B_j$ it holds $\omega_1 > \omega_2$.
  \end{enumerate}
  Since $\mu(\{\omega\})=0$ for all $\omega \in \Omega$, such partition always exists.

  Define $D_{i,d} = \{\omega \in B_i \mid T^{\circ l}(\omega) \notin \bigcup_{j=i}^{\infty}B_j \text{ for all } l = 1,...,d-1\}$.
  It holds
  \begin {equation}\label{ZeroMeasure}
     \bigcup_{l=1}^{d-1} \left( D_{i,d} \cap T^{- \circ l}(\bigcup_{j=i}^{\infty}B_j) \right) = \emptyset.
  \end {equation}
  For all $d \in {\mathbb N}$, \eqref{ZeroMeasure} provides $\widetilde{V}^L_d \subseteq \bigcup_{i=1}^{\infty} D_{i,d}$ and, since $D_{i,d} \subseteq B_i$, it holds
  \begin{align}\label{MeasureBound}
    \mu(\widetilde{V}^L_d) &\leq \mu(\bigcup_{i=1}^{\infty} D_{i,d}) = \sum_{i=1}^{\infty}\mu( D_{i,d}) \leq \sum_{i=1}^{k}\mu( D_{i,d}) + \sum_{i=k+1}^{\infty}\mu( B_i)  \nonumber \\
		       &\leq \sum_{i=1}^{k}\mu( D_{i,d}) + \sum_{i=k+1}^{\infty}\frac{\mu(A)}{2^i}  \leq \sum_{i=1}^{k}\mu( D_{i,d}) + \frac{\mu(A)}{2^k}
  \end{align}
  for all $k \in {\mathbb N}$. On the other hand, by the ergodicity of $T$ (compare \cite{Walters82}) and by \eqref{ZeroMeasure} we have
  \begin {equation*}
    \mu(\bigcap_{d=1}^{\infty}D_{i,d})\, \mu(\bigcup_{j=i}^{\infty}B_j) = \lim_{m \to \infty}\frac{1}{m} \sum_{l=1}^{m-1}  \mu\left(\bigcap_{d=1}^{\infty}D_{i,d} \cap T^{- \circ l}(\bigcup_{j=i}^{\infty}B_j)\right) = 0.
  \end {equation*}
  Therefore, $\mu(\bigcup_{j=i}^{\infty}B_j)>0$ implies $\mu(\bigcap_{d=1}^{\infty}D_{i,d}) = 0$ and, since $D_{i,1} \supseteq D_{i,2} \supseteq \ldots$, it holds
  \begin{equation*}
    \lim_{d \to \infty}\mu( D_{i,d}) = \mu(\bigcap_{d=1}^{\infty}D_{i,d}) = 0
  \end{equation*}
  for all $i\in {\mathbb N}$.

Now let $\varepsilon>0$. Fix some $k\in {\mathbb N}$ with $k>\log_2\frac{4}{\varepsilon}$ and $d_\varepsilon$ with $\mu( D_{i,d})<\frac{\varepsilon}{4k}$ for all $i=1,2,\ldots ,k$ and $d>d_\varepsilon$. Then, owing to \eqref{MeasureBound},
for $d > d_\varepsilon$ it holds
\begin {equation*}
\mu(\widetilde{V}^L_d) < k\,\frac{\varepsilon}{4k}+\frac{\varepsilon}{4}=\frac{\varepsilon}{2}
\end {equation*}
completing the proof.\hfill $\Box$
\end{proof}

Now we are coming to the proof of Theorem \ref{Theorem1}. Given $\varepsilon>0$, let $r>\frac{3}{\varepsilon}$ and let $\{A_i\}_{i=1}^r$ be a partition of $\Omega$ into intervals $A_i$ with $\mu(A_i)=\frac{1}{r}$. Furthermore, fix some $d_\varepsilon\in {\mathbb N}$ with
\begin{equation}\label{l1}
\mu( A_i \cap T^{-\circ d}(A_i) ) \leq \mu^2(A_i)+\frac{\varepsilon}{3r}=\frac{1}{r^2}+\frac{\varepsilon}{3r}
\end{equation}
and
\begin{equation}\label{l2}
\mu(\widetilde{V}_d)\leq\frac{\varepsilon}{3r}
\end{equation}
for all $i=1,2,\ldots ,r$ and all $d>d_\varepsilon$,
which is possible by the strong-mixing of $T$ and by Lemma \ref{Lemma1}, respectively.

For $\omega\in V_d\cap A_i$ it is impossible that both $T^{\circ d}(\omega)\not\in A_i$ and $\omega\not\in \widetilde{V}_d(A_i)$, implying
\begin{eqnarray*}
V_d = \bigcup_{i=1}^r\, (V_d\cap A_i)&\subset& \bigcup_{i=1}^r\, ((A_i\cap T^{-\circ d}(A_i))\cup\widetilde{V}_d(A_i))\\
&=& \bigcup_{i=1}^r\, (A_i\cap T^{-\circ d}(A_i))\cup\bigcup_{i=1}^r\, \widetilde{V}_d(A_i).
\end{eqnarray*}
From this, \eqref{l1}, and \eqref{l2}, one obtains
\begin{eqnarray*}
\mu(V_d)&\leq&\sum_{i=1}^r\, \mu(A_i\cap T^{-\circ d}(A_i))+\sum_{i=1}^r\, \mu(\widetilde{V}_d(A_i))\\
&\leq&r\left (\frac{1}{r^2}+\frac{\varepsilon}{3r}\right )+\frac{\varepsilon}{3}<\varepsilon.
\end{eqnarray*}

\begin{remark}
The technical assumption that $\mu (\{\omega\})=0$ for all $\omega\in\Omega$ is rather weak. In the ergodic (resp.~strong-mixing) case, $\mu (\{\omega\})>0$ would imply that $\omega$ is a periodic (resp.~fixed) point and that $\mu$ is concentrated on the orbit of $\omega$ (resp.~on $\omega$).
\end{remark}

This work was supported by the Graduate School for Computing in Medicine and Life Sciences
funded by Germany's Excellence Initiative [DFG GSC 235/1].

\end{document}